\documentclass[conference]{IEEEtran}
\usepackage{amsmath,amssymb,amsthm,cite,graphicx,array}
\usepackage{graphicx}
\usepackage{float}
\usepackage{caption}
\usepackage{multicol}
\usepackage{mathrsfs}
\usepackage{amsmath}
\usepackage{amssymb}
\usepackage{amsthm}
\usepackage{multicol}
\usepackage{algorithm}
\usepackage[noend]{algpseudocode}
\usepackage{float}
\usepackage{placeins}
\usepackage{epstopdf}
\usepackage{multirow}
\usepackage{subfigure}
\usepackage{enumitem}
\usepackage[utf8]{inputenc}
\floatstyle{ruled}
\newfloat{algorithm}{tbp}{loa}
\providecommand{\algorithmname}{Algorithm}
\floatname{algorithm}{\protect\algorithmname}
\theoremstyle{remark}
\newtheorem{theorem}{Theorem}
\newtheorem{corollary}{Corollary}
\newtheorem{lemma}{Lemma}
\newtheorem{definition}{Definition}
\newtheorem{remark}{Remark}
\newtheorem{note}{Note}
\theoremstyle{remark}
\newtheorem{example}{Example}
\ifCLASSINFOpdf
\else
\fi

\title{A Relation Between Weight Enumerating Function and Number of Full Rank Sub-matrices}

\begin{document}

\author{Mahesh~Babu~Vaddi~and~B.~Sundar~Rajan\\ 
 Department of Electrical Communication Engineering, Indian Institute of Science, Bengaluru 560012, KA, India \\ E-mail:~\{vaddi,~bsrajan\}@iisc.ac.in }
 
\maketitle
\begin{abstract}
In most of the network coding problems with $k$ messages, the existence of binary network coding solution over $\mathbb{F}_2$ depends on the existence of adequate sets of $k$-dimensional binary vectors such that each set comprises of linearly independent vectors. In a given $k \times n$ ($n \geq k$) binary matrix, there exist ${n}\choose{k}$ binary sub-matrices of size $k \times k$.  Every possible $k \times k$ sub-matrix may be of full rank or singular depending on the columns present in the matrix. In this work, for full rank binary matrix $\mathbf{G}$ of size $k \times n$ satisfying certain condition on minimum Hamming weight, we establish a relation between the number of full rank sub-matrices of size $k \times k$ and the weight enumerating function of the error correcting code with $\mathbf{G}$ as the generator matrix. We give an algorithm to compute the number of full rank $k \times k$ submatrices.
\end{abstract}

\section{Introduction}
\IEEEPARstart{N}etwork coding is a technique to enhance the rate of information transmission in a network by taking advantage of redundancy in demands of various receivers. Network coding was introduced in \cite{NC1}. An acyclic network can be represented as an acyclic directed graph $D=(V,E)$, where $V$ is the set of all nodes and $E$ is the set of all edges in the network. Every edge $E$ in $D$ can carry one message symbol from the given finite field $\mathbb{F}_q$. There exists a unique node $S$ in $D$, called the source node, the source node has $k$ message symbols $x_i \in \mathbb{F}_q$ for $i \in [1:k]$. There exist some receiver nodes in the network and each receiver wants some subset of message symbols $\{x_1,x_2,\ldots,x_k\}$. The objective in network coding is to satisfy the demands of all receivers by minimizing the number of transmissions in the edges of the network. 

A solution to the network coding problem is the design of $k$-dimensional precoding vectors (these precoding vectors are called global encoding kernels in network coding) to each edge in the network such that these precoding vectors satisfy some linear independence constraints. An $k$-dimensional linear network code on an acyclic network over a field $\mathbb{F}_q$ consists of a global encoding mapping $f_{n_i,n_j}:\mathbb{F}_q^{k}\rightarrow \mathbb{F}_q$ for each edge $(n_i,n_j) \in E$ in the network. The existence of linear network coding solutions was studied in \cite{NC2} and \cite{NC3}. A construction of linear network coding for multicast network coding problems was studied in \cite{NC4}. Depending on the network topology and requirement of the receivers, a network may or may not have a solution in $\mathbb{F}_2$. In most of the network coding problems, the existence of binary network coding solution over $\mathbb{F}_2$ depends on the existence of adequate sets of $k$-dimensional binary vectors such that each set comprises of linearly independent vectors.


The problem of index coding with side-information was introduced by Birk and Kol \cite{ISCO}. In index coding problems, we often require the condition that a specific number of $k \times k$ submatrices in a given $k \times n$ matrix needs to have full rank. In index coding, the default choice for this $k \times n$ matrix is an $k \times n$ Vandermonde matrix. The Vandermonde matrix exists over higher fields and the field size required depends on $n$. Hence, it is useful to design binary $k \times n$ matrices satisfying the given linear independence conditions.

\subsection{Motivating Example}
Consider the network coding problem as shown in Fig. \ref{fig11}. In this network coding problem, the source comprises of two messages and there exist six receivers, each of the receiver wants both the messages. For this network, linear solution over $\mathbb{F}_2$ is not possible \cite{RY}. This is known as the ${4}\choose{2}$ combinational network. 

However, for the network coding problem shown in Fig. \ref{fig11}, if we remove any one receiver out of six receivers, then linear solution over $\mathbb{F}_2$ is possible as shown in Table \ref{table1}. This follows from the fact that for any binary matrix of size $4 \times 2$, there exist atmost $5$ full rank submatrices of size $2 \times 2$ and this also means that binary coding solution does not exist for the ${4}\choose{2}$ combinational network. Hence, the number of full rank $k \times k$ submatrices in a binary matrix of size $k \times n$ is useful in analysing some multicast network coding problems.
\begin{figure}[h]
\centering
\includegraphics[scale=0.4]{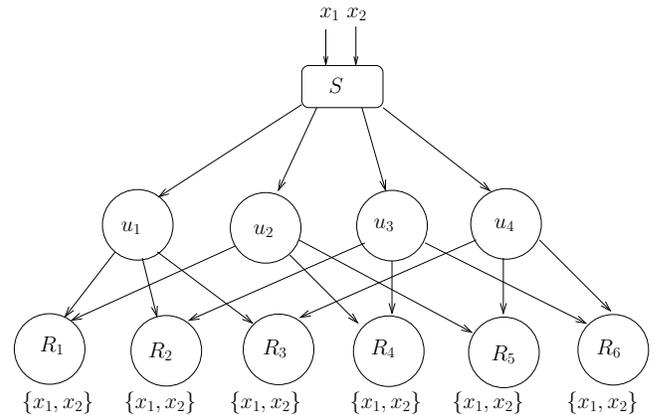}
\caption{An instance of network coding problem with no linear binary solution.}
\label{fig11}
\end{figure}
\begin{table}[h]
\centering
\setlength\extrarowheight{2pt}
\begin{tabular}{|c|c|}
\hline
Receivers Present & Global encoding kernels for $\text{Out}(S)$ \\
\hline
$R_1,R_2,R_3,R_4,R_5$ & $f_{s,u_1}=[1~0]^\mathsf{T},f_{s,u_2}=[0~1]^\mathsf{T}$ \\
~ & $f_{s,u_3}=[1~1]^\mathsf{T},f_{s,u_4}=[1~1]^\mathsf{T}$ \\
\hline
$R_1,R_2,R_3,R_4,R_6$ & $f_{s,u_1}=[1~0]^\mathsf{T},f_{s,u_2}=[1~1]^\mathsf{T}$ \\
~ & $f_{s,u_3}=[0~1]^\mathsf{T},f_{s,u_4}=[1~1]^\mathsf{T}$ \\
\hline
$R_1,R_2,R_3,R_5,R_6$ & $f_{s,u_1}=[1~0]^\mathsf{T},f_{s,u_2}=[1~1]^\mathsf{T}$ \\
~ & $f_{s,u_3}=[1~1]^\mathsf{T},f_{s,u_4}=[0~1]^\mathsf{T}$ \\
\hline
$R_1,R_2,R_4,R_5,R_6$ & $f_{s,u_1}=[1~1]^\mathsf{T},f_{s,u_2}=[1~0]^\mathsf{T}$ \\
~ & $f_{s,u_3}=[0~1]^\mathsf{T},f_{s,u_4}=[1~1]^\mathsf{T}$ \\
\hline
$R_1,R_3,R_4,R_5,R_6$ & $f_{s,u_1}=[1~1]^\mathsf{T},f_{s,u_2}=[1~0]^\mathsf{T}$ \\
~ & $f_{s,u_3}=[1~1]^\mathsf{T},f_{s,u_4}=[0~1]^\mathsf{T}$ \\
\hline
$R_2,R_3,R_4,R_5,R_1$ & $f_{s,u_1}=[1~1]^\mathsf{T},f_{s,u_2}=[1~1]^\mathsf{T}$ \\
~ & $f_{s,u_3}=[1~0]^\mathsf{T},f_{s,u_4}=[0~1]^\mathsf{T}$ \\
\hline
\end{tabular}
\vspace{5pt}
\caption{Global encoding kernels at $\text{Out}(S)$ for the six different network coding problems obtained after removing one receiver $R_i$ for $i \in [1:6]$ in Fig. \ref{fig11}. Global encoding kernels for the edges in $\text{Out}(u_i)$ is equal to $f_{s,u_i}$ for $i\in [1:4].$}
\label{table1}
\vspace{-5pt}
\end{table}

A linear $(n,k)$ error correcting code over the field $\mathbb{F}_q$ is a $k$-dimensional subspace of $\mathbb{F}_q^n$. Let $\mathfrak{C}$ denote the $(n,k)$ linear error correcting code. Let $\mathbf{G}$ and $\mathbf{H}$ be the generator and parity check matrix of $\mathfrak{C}$. The weight enumerating function of $\mathfrak{C}$ is given by the polynomial 
\begin{align*}
W_{\mathfrak{C}}(x,y)=\sum_{d=0}^{n}A_d x^{n-d} y^d,
\end{align*}
where $A_d$ is the number of codewords with Hamming weight $d$ and $x,y$ are indeterminates. We refer $A_d$s for $d \in [0:n]$ as weight enumerator coefficients (WECs). The linear block code generated by $\mathbf{H}$ as generator matrix is called dual of $\mathfrak{C}$ and is denoted by $\mathfrak{C}^\mathsf{T}$. The relation between weight enumerator $W_{\mathfrak{C}}(x,y)$ of $\mathfrak{C}$ and  weight enumerator $W_{\mathfrak{C}^\mathsf{T}}(x,y)$ of $\mathfrak{C}^\mathsf{T}$ was derived by MacWilliams in \cite{ECC2}. $W_{\mathfrak{C}}(x,y)$ and $W_{\mathfrak{C}^\mathsf{T}}(x,y)$ are related as
\begin{align}
\label{wedual}
W_{\mathfrak{C}^\mathsf{T}}(x,y)=\frac{1}{\vert \mathfrak{C} \vert} W_{\mathfrak{C}}(x+y,x-y).
\end{align}

For more details on error correcting codes and details on weight enumerators, the readers are referred to \cite{ECC}.

Given a subspace $\mathbf{V}$ of $\mathbb{F}^n_2$, the space of all vectors orthogonal to $\mathbf{V}$ in $\mathbb{F}^n_2$ is called orthogonal complement of $\mathbf{V}$. For any given arbitrary full rank matrix $\mathbf{G}$ of size $k \times n$, from the fundamental theorem of linear algebra, the null space is the orthogonal complement of the row space. The dimension of the row space of $\mathbf{G}$ is $k$, the dimension of null space of $\mathbf{G}$ is $n-k$ and the sum of the dimension of the null space and row space is equal to $n$. For any given arbitrary full rank matrix $\mathbf{A}$ of size $k \times n$, we refer $\mathbf{B}$ of size $(n-k) \times n$ as orthogonal complement of $\mathbf{A}$ if the rank of $\mathbf{B}$ is $n-k$ and $\mathbf{A}\mathbf{B}^{\mathsf{T}}$ is all zero matrix. 

Let $\mathbf{G}$ be a full rank $k \times n$ binary matrix. Let $\mathbf{H}$ be an orthogonal complement of $\mathbf{G}$. If $k \geq n-k$, let $d^*$ be the minimum Hamming weight of the error correcting code generated by $\mathbf{H}$. Else, let $d^*$ be the minimum Hamming weight of the error correcting code generated by $\mathbf{G}$.

\subsection{Contributions}
The contributions of this paper are summarized below:
\begin{itemize}
\item For a given binary full rank $k \times n$ matrix $\mathbf{G}$, we show that the number of full rank $k \times k$ submatrices in $\mathbf{G}$ are equal to the number of full rank $(n-k) \times (n-k)$ submatrices in the orthogonal complement of $\mathbf{G}$.
\item In any binary full rank matrix $\mathbf{G}$ of size $k \times n$, if $\frac{3d^*}{2} > \text{max}(k,n-k)$, we establish a relation between the weight enumerating function $W_{\mathfrak{C}}(x,y)$ of $\mathfrak{C}$ generated by $\mathbf{G}$ and the number of full rank submatrices of size $k \times k$ in $\mathbf{G}$. We explicitly quantify the number of full rank submatrices of size $k \times k$  in $\mathbf{G}$.
\item For a binary full rank matrix $\mathbf{G}$ of size $k \times n$ satisfying the condition $\frac{3d^*}{2} > \text{max}(k,n-k)$, we give an algorithm to compute the number of full rank submatrices of size $k \times k$ in $\mathbf{G}$.
\end{itemize}


\section{On weight enumerator function and the number of full-rank matrices}
\label{sec2}
Lemma \ref{lemma1} and Lemma \ref{lemma2} given below are useful to derive the main result in this section.
\begin{lemma}
\label{lemma1}
Let $\mathbf{G}$ be a binary $k \times n$ matrix. Let $\mathbf{G}^{(s)}$ be the $k \times n$ symmetric matrix obtained from $\mathbf{G}$ by using elementary row operations. The matrix $\mathbf{G}^{(s)}$ is of the form $[\mathbf{I}_{k \times k}: \mathbf{P}_{k \times (n-k)}]$. If the $k$ columns corresponding to the indices $\{i_1,i_2,\ldots,i_k\} \subset [1:n]$ are linearly dependent in $\mathbf{G}$, then the $k$ columns corresponding to the indices $\{i_1,i_2,\ldots,i_k\}$ are linearly dependent in $\mathbf{G}^{(s)}$ also. 
\end{lemma}
\begin{proof}
Let $r_i^\prime$ and $r_i$ be the $i$th row of $\mathbf{G}$ and $\mathbf{G}^{(s)}$ respectively for every $i \in [1:k]$. Let $c_j^\prime$ and $c_j$ be the $j$th column of $\mathbf{G}$ and $\mathbf{G}^{(s)}$ respectively for every $j \in [1:n]$. Let $g_{i,j}^\prime$ and $g_{i,j}$ be the element in the $i$th row and $j$th column of $\mathbf{G}$ and $\mathbf{G}^{(s)}$ respectively for every $i \in [1:k]$ and $j \in [1:n]$. In the matrix $\mathbf{G}^{(s)}$, let the $i$th row be a linear combination of $k$ rows in the matrix $\mathbf{G}$. We can write 
\begin{align}
\label{lemma1eq1}
r_i=f_i(r_1^\prime,r_2^\prime,\ldots,r_k^\prime)~\text{for}~i \in [1:k], 
\end{align}
where $f_i$ is some linear function over $\mathbb{F}_2$. From \eqref{lemma1eq1}, we can write $g_{i,j}$ as
\begin{align}
\label{lemma1eq11}
g_{i,j}=f_i(g_{1,j}^\prime,g_{2,j}^\prime,\ldots,g_{k,j}^\prime)~\text{for}~i \in [1:k]~\text{and}~j \in [1:n].
\end{align}
Given that the $k$ columns corresponding to the indices $\{i_1,i_2,\ldots,i_k\} \subset [1:n]$ are linearly dependent in $\mathbf{G}$. There exists $a_1,a_2,\ldots,a_k \in \mathbb{F}_2$ such that 
\begin{align}
\label{lemma1eq2}
a_1 c_{i_1}^\prime+a_2 c_{i_2}^\prime+\ldots+a_k c_{i_k}^\prime=0.
\end{align}
From \eqref{lemma1eq2}, we have
\begin{align}
\label{lemma1eq3}
\nonumber
&a_1 g_{1,i_1}^\prime+a_2 g_{1,i_2}^\prime+\ldots+a_k g_{1,i_k}^\prime=0 \\&
\nonumber
a_1 g_{2,i_1}^\prime+a_2 g_{2,i_2}^\prime+\ldots+a_k g_{2,i_k}^\prime=0 \\&
\nonumber
~~~~~~~~~~~~~\vdots~~~~~~~~~~~~~~~~\vdots \\&
a_1 g_{k,i_1}^\prime+a_2 g_{k,i_2}^\prime+\ldots+a_k g_{k,i_k}^\prime=0.
\end{align}
From \eqref{lemma1eq3}, for every $i \in [1:k]$, we have 
\begin{align}
\label{lemma1eq4}
\nonumber
&a_1f_i(g_{1,i_1}^\prime,g_{2,i_1}^\prime,\ldots,g_{k,i_1}^\prime)+\\&
\nonumber
a_2f_i(g_{1,i_2}^\prime,g_{2,i_2}^\prime,\ldots,g_{k,i_2}^\prime)+\ldots+\\&
a_kf_i(g_{1,i_k}^\prime,g_{2,i_k}^\prime,\ldots,g_{k,i_k}^\prime)=0.
\end{align}
From \eqref{lemma1eq4}, we have
\begin{align}
\label{lemma1eq5}
a_1g_{i,i_1}+a_2g_{i,i_2}+\ldots+a_k g_{i,i_k}=0~\text{for~every~}i \in [1:k].
\end{align}
From \eqref{lemma1eq5}, we have
\begin{align*}
a_1 c_{i_1}+a_2 c_{i_2}+\ldots+a_k c_{i_k}=0.
\end{align*}
That is, the columns $c_{i_1},c_{i_2},\ldots,c_{i_k}$ of $\mathbf{G}^{(s)}$ are linearly dependent.
\end{proof}

\begin{lemma}
\label{lemma2}
Let $\mathbf{G}$ be a binary $k \times n$ matrix of the form $[\mathbf{I}_{k \times k}: \mathbf{P}_{k \times (n-k)}]$. Define the $(n-k) \times n$ matrix $\mathbf{H}$ as given below
\begin{align*}
\mathbf{H}=[\mathbf{P}_{k \times (n-k)}^{\mathsf{T}}: \mathbf{I}_{(n-k) \times (n-k)}].
\end{align*}
If $\{i_1,i_2,\ldots,i_k\} \subset [1:n]$ be indices of any $k$ linearly dependent columns in $\mathbf{G}$, then the $(n-k)$ columns corresponding to the indices $[1:n]/ \{i_1,i_2,\ldots,i_k\}$ in $\mathbf{H}$ are linearly dependent. 
\end{lemma}
\begin{proof}
We have 
\begin{align*}
\mathbf{G}\mathbf{H}^{\mathsf{T}}&=[\mathbf{I}_{k \times k}: \mathbf{P}_{k \times (n-k)}][\mathbf{P}_{k \times (n-k)}^{\mathsf{T}}: \mathbf{I}_{(n-k) \times (n-k)}]^{\mathsf{T}}\\&=\mathbf{P}_{k \times (n-k)}+\mathbf{P}_{k \times (n-k)}=\mathbf{0}_{k \times (n-k)}.
\end{align*}

The rank of the matrices $\mathbf{G}$ and $\mathbf{H}$ are $k$ and $n-k$ respectively and every row in $\mathbf{G}$ is orthogonal (inner product is zero) to every other row in $\mathbf{H}$ over $\mathbb{F}_2$. Hence, the row space spanned by $\mathbf{G}$ and $\mathbf{H}$ are orthogonal complements. Let $\{i_1,i_2,\ldots,i_k\} \subset [1:n]$ be indices of any $k$ linearly dependent columns in $\mathbf{G}$. Consider a binary vector $\mathbf{v} \in \mathbb{F}_2^n$ such that $\mathbf{v}$ has $1$s in $i_1,i_2,\ldots,i_k$ positions and $0$ in other positions. The vector $\mathbf{v}$ must be in the null space of $\mathbf{G}$ and row space of $\mathbf{H}$. Consider a binary vector $\tilde{\mathbf{v}} \in \mathbb{F}_2^n$ such that it has $0$ in $i_1,i_2,\ldots,i_k$ positions and $1$ in other positions. The vector $\tilde{\mathbf{v}}$ must be in the row space of $\mathbf{G}$ (row space of $\mathbf{G}$ and $\mathbf{H}$ are orthogonal complements). Hence, we have $\mathbf{H}\tilde{\mathbf{v}}=0$ and the columns corresponding to the indices $[1:n]/ \{i_1,i_2,\ldots,i_k\}$ in $\mathbf{H}$ are linearly dependent. 
\end{proof}
\begin{corollary}
\label{lemma3}
Let $\mathbf{G}$ be a binary $k \times n$ matrix of the form $[\mathbf{I}_{k \times k}: \mathbf{P}_{k \times (n-k)}]$. Define the $(n-k) \times n$ matrix $\mathbf{H}$ as given below
\begin{align*}
\mathbf{H}=[\mathbf{P}_{k \times (n-k)}^{\mathsf{T}}: \mathbf{I}_{(n-k) \times (n-k)}].
\end{align*}
If $\{i_1,i_2,\ldots,i_k\} \subset [1:n]$ be indices of any $k$ linearly independent columns in $\mathbf{G}$, then the $(n-k)$ columns corresponding to the indices $[1:n]/ \{i_1,i_2,\ldots,i_k\}$ in $\mathbf{H}$ are linearly independent. 
\end{corollary}
\begin{remark}
The number of linearly independent submatrices of size $k \times k$ in $\mathbf{G}$ is equal to the number of linearly independent submatrices of size $(n-k) \times (n-k)$ in $\mathbf{H}$.
\end{remark}
\begin{example}
\label{ex2}
Consider the matrices $\mathbf{G}$ and $\mathbf{H}$ given below.

\arraycolsep=0.0pt
\setlength\extrarowheight{-4.0pt}
{
$$\mathbf{G}=\left[\begin{array}{*{20}c}
   1~0~0~0~1~1~1  \\
   0~1~0~0~1~1~0 \\
   0~0~1~0~1~0~1 \\
   0~0~0~1~0~1~1

   \end{array}\right], \mathbf{H}=\left[\begin{array}{*{20}c}
   1~1~1~0~1~0~0  \\
   1~1~0~1~0~1~0 \\
   1~0~1~1~0~0~1 \\

   \end{array}\right].$$
   } 
   
The columns with indices $\{1,2,3,5\}$ are linearly dependent in $\mathbf{G}$. Hence, the columns $[1:7]\setminus \{1,2,3,5\}=\{4,6,7\}$ are linearly dependent in $\mathbf{H}$. 

The columns with indices $\{1,4,6,7\}$ are linearly independent in $\mathbf{G}$. Hence, the columns $[1:7]\setminus \{1,4,6,7\}=\{2,3,5\}$ are linearly independent in $\mathbf{H}$.
\end{example}

\begin{definition}
Let $\mathbf{G}=[\mathbf{I}_{k \times k}: \mathbf{P}_{k \times (n-k)}]$ be a binary full rank matrix of size $k \times n$. We can select any set of $t$ columns for $1 \leq t \leq n-k$ from the columns of $\mathbf{P}_{k \times (n-k)}$ in $2^{n-k}-1$ ways. Let $\{c_{i_1},c_{i_2},\ldots,c_{i_t}\}$ be the $t$ selected columns in $\mathbf{P}_{k \times (n-k)}$. The effective distance $d^{(e)}(i_1,i_2,\ldots,i_t)$ of these $t$ columns is defined as
\begin{align*}
d^{(e)}(i_1,i_2,\ldots,i_t)=d_H(c_{i_1}+c_{i_2}+\ldots+c_{i_t})+t,
\end{align*}
where $d_H(c_{i_1}+c_{i_2}+\ldots+c_{i_t})$ is the Hamming weight of $(n-k)$ dimensional binary vector $c_{i_1}+c_{i_2}+\ldots+c_{i_t}$.
\end{definition}
\begin{example}
\label{ex21}
Consider the matrix $\mathbf{G}$ given below.

\arraycolsep=0.0pt
\setlength\extrarowheight{-4.0pt}
{
$$\mathbf{G}=\left[\begin{array}{*{20}c}
   1~0~0~1~1~1  \\
   0~1~0~1~1~0 \\
   0~0~1~1~0~1 \\

   \end{array}\right].$$
   } 

The effective distances of the matrix $\mathbf{G}$ are as follows:
\begin{align*}
&d^{(e)}(4)=d_H(111)+1=4,\\& d^{(e)}(5)=d_H(110)+1=3,\\&
d^{(e)}(6)=d_H(101)+1=3, \\& d^{(e)}(4,5)=d_H(111+110)+2=3,\\&
d^{(e)}(4,6)=d_H(111+101)+2=3, \\& d^{(e)}(5,6)=d_H(101+110)+2=4,\\&
d^{(e)}(4,5,6)=d_H(111+110+101)+3=4.
\end{align*}
\end{example}

Theorem \ref{s1thm1} establishes the relation between weight enumerating function and the number of $k \times k$ singular matrices of a full rank binary matrix of size $k \times n$.
\begin{theorem}
\label{s1thm1}
Consider an arbitrary binary full rank matrix $\mathbf{G}$ of size $k \times n~(n \geq k)$. Let $\mathfrak{C}$ be the linear block code generated by $\mathbf{G}$ as generator matrix and $\mathfrak{C}^\mathsf{T}$ be the dual code of $\mathfrak{C}$. Let $\mathcal{D}$ be the number of singular submatrices of size $k \times k$ in $\mathbf{G}$. If 
$\frac{3d^*}{2} > \text{max}(k,n-k)$, then 
\begin{align}
\label{s1main}
\mathcal{D} = \sum_{d=d^*}^kA_d{{n-d}\choose{n-k}},
\end{align}
where 
\begin{itemize}
\item $A_d$s are the weight enumerating coefficients of $W_{\mathfrak{C}^\mathsf{T}}(x,y)$ if $k \geq n-k$
\item $A_d$s are the weight enumerating coefficients of $W_{\mathfrak{C}}(x,y)$ if  $k < n-k$.
\end{itemize}
\end{theorem}
\begin{proof}
Let $\mathbf{G}^{(s)}=[\mathbf{I}_{k \times k}~\mathbf{P}_{k \times (n-k)}]$ be the systematic generator matrix for $\mathfrak{C}$. From Lemma \ref{lemma1}, the number of linearly dependent $k$-column sets in $\mathbf{G}$ is equal to the number of linearly dependent $k$-column sets in $\mathbf{G}^{(s)}$. Let $\mathbf{H}=[\mathbf{P}_{n-k \times k}^\mathsf{T}~\mathbf{I}_{(n-k) \times (n-k)}]$ be the parity check matrix for $\mathfrak{C}$. Let $c_j$ be the $j$th column of $\mathbf{G}^{(s)}$ for $j \in [1:n]$.

\textbf{Case (i)} $k \geq n-k$:-

For every $\{i_1,i_2,\ldots,i_t\} \subseteq [k+1:n]$, let $j_1,j_2,\ldots$, $j_{d_H(c_{i_1}+c_{i_2}+\ldots+c_{i_t})}$ be the positions of 1s present in $c_{i_1}+c_{i_2}+\ldots+c_{i_t}$. Note that $c_{j_1},c_{j_2},\ldots,c_{j_{d_H(c_{i_1}+c_{i_2}+\ldots+c_{i_t})}}$ are the columns of identity matrix present in $\mathbf{G}^{(s)}$ with $1$s in $j_1,j_2,\ldots$, $j_{d_H(c_{i_1}+c_{i_2}+\ldots+c_{i_t})}$ positions respectively. Hence, in the matrix $\mathbf{G}$, the $d^{(e)}(i_1,i_2,\ldots,i_t)$ number of columns $$\underbrace{\{c_{i_1},c_{i_2},\ldots,c_{i_t}\}}_{t~\text{columns}} \cup \underbrace{\{c_{j_1},c_{j_2},\ldots,c_{j_{d_H(c_{i_1}+c_{i_2}+\ldots,c_{i_t})}}\}}_{d_H(c_{i_1}+c_{i_2}+\ldots+c_{i_t})}$$ are linearly dependent. 

If $d^{(e)}(i_1,i_2,\ldots,i_t)\leq k$, then, any combination of $k-d^{(e)}(i_1,i_2,\ldots,i_t)$ columns in the remaining $n-d^{(e)}(i_1,i_2,\ldots,i_t)$ along with the $d^{(e)}(i_1,i_2,\ldots,i_t)$ columns $\{c_{i_1},c_{i_2},\ldots,c_{i_t}\} \cup \{c_{j_1},c_{j_2},\ldots,c_{j_{d_H(c_{i_1}+c_{i_2}+\ldots,c_{i_t})}}\} $of $\mathbf{G}$ are linearly dependent. Hence, there exists ${n-d^{(e)}(i_1,i_2,\ldots,i_t)}\choose{k-d^{(e)}(i_1,i_2,\ldots,i_t})$ sets of size $k$, which are linearly dependent and which comprise $\{c_{i_1},c_{i_2},\ldots,c_{i_t}\}$. The total number of linearly dependent sets of $k$-columns are 
\begin{align}
\label{dbound}
\nonumber
\mathcal{D}&=\sum_{\forall \{i_1,i_2,\ldots,i_t\}\subseteq [k+1:n]} {{n-d^{(e)}(i_1,i_2,\ldots,i_t)}\choose{k-d^{(e)}(i_1,i_2,\ldots,i_t)}} \\&
=\sum_{\forall \{i_1,i_2,\ldots,i_t\}\subseteq [k+1:n]} {{n-d^{(e)}(i_1,i_2,\ldots,i_t)}\choose{n-k}}.
\end{align}

We have 
\begin{align}
\label{wec5}
\nonumber
d^{(e)}(i_1,i_2,\ldots,i_t)&=d_H(c_{i_1}+c_{i_2}+\ldots+c_{i_t})+t \\&
\nonumber
=d_H(c_{i_1}+c_{i_2}+\ldots+c_{i_t})\\&
\nonumber
~~~+d_H(\text{sum of any t columns of}~\mathbf{I}_{n-k})\\&
\nonumber
=\text{Hamming weight of codeword generated}\\& ~~~~\text{by}~\underbrace{i_1-k,i_2-k,\ldots,i_t-k}_{\in [1:n-k]}~\text{rows of}~\mathbf{H}.
\end{align}


By using \eqref{wec5}, we can write \eqref{dbound} as 
\begin{align}
\label{dbound1}
\nonumber
\mathcal{D}&=\sum_{\forall \{i_1,i_2,\ldots,i_t\}\subseteq [k+1:n]} {{n-d^{(e)}(i_1,i_2,\ldots,i_t)}\choose{n-k}}\\&=\sum_{d=d^*}^{n} A_d{{n-d}\choose{n-k}}=\sum_{d=d^*}^{k} A_d{{n-d}\choose{n-k}},
\end{align}
where $A_d$s are the weight enumerating coefficients of $\mathfrak{C}^\mathsf{T}$. The last equality in \eqref{dbound1} follows from the fact that the codewords with weight from $k+1$ to $n$ in $\mathbf{H}$ do not yield any set of $k$ columns which are linearly dependent.

Now, if $\frac{3d^*}{2} > \text{max}(k,n-k)=k$, we would show that there is no double counting in \eqref{dbound}. From \eqref{wec5}, the minimum value of $d^{(e)}(i_1,i_2,\ldots,i_t)$ for any $\{i_1,i_2,\ldots,i_t\} \subseteq [k+1:n]$ is $d^*$. Consider a codeword of $\mathbf{H}$ with Hamming weight $d_1 \geq d^*$ and having $1$s in $i_1,i_2,\ldots,i_{d_1}$ positions. Hence, the $i_1,i_2,\ldots,i_{d_1}$th columns with any combination of $(k-d_1)$ columns in the remaining $(n-d_1)$ columns of $\mathbf{G}$ are linearly dependent. There exists ${n-d_1}\choose{k-d_1}$ sets of size $k$, which are linearly dependent and which comprise $i_1,i_2,\ldots,i_{d_1}$ columns. Let these ${n-d_1}\choose{k-d_1}$ sets of size $k$ columns be $S_1$. Consider any other codeword of $\mathbf{H}$ with Hamming weight $d_2 \geq d^*$ and having $1$s in If $k_1,k_2,\ldots,k_{d_2}$ positions.  Hence, the $k_1,k_2,\ldots,k_{d_2}$th columns with any combination of $(k-d_2)$ columns in the remaining $(n-d_2)$ columns of $\mathbf{G}$ are linearly dependent. There exists ${n-d_2}\choose{k-d_2}$ sets of size $k$, which are linearly dependent and which comprise $i_1,i_2,\ldots,i_{d_2}$ columns. Let these ${n-d_2}\choose{k-d_2}$ sets of size $k$ columns be $S_2$. We need to show all the $k$ column sets in $S_1$ are distinct to the $k$ columns sets in $S_2$. 

\begin{figure}
\centering
\includegraphics[scale=0.45]{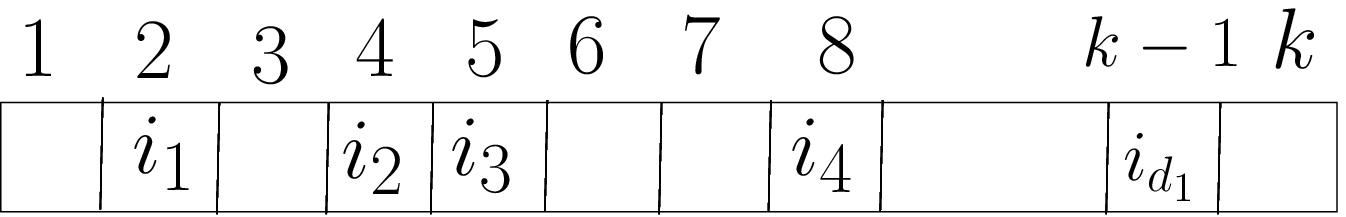}\\
\caption{}
\label{fig1}
\end{figure}

\begin{figure}
\centering
\includegraphics[scale=0.45]{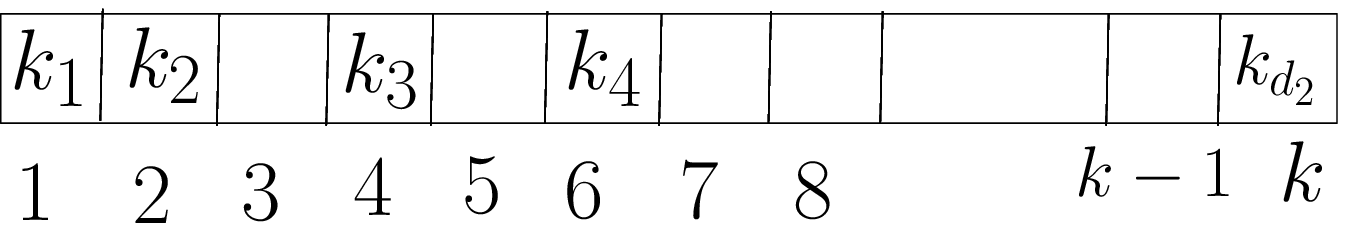}\\
\caption{}
\label{fig2}
\end{figure}

Let there exists a $k$ column set which is common in both $S_1$ and $S_2$. With out loss of generality, we assume that this common set is as shown in Figure \ref{fig1} and Figure \ref{fig2}. Because the $k$ indices in Figure \ref{fig1} and Figure \ref{fig2} are equal, the $k-d_1$ vacant positions in Figure \ref{fig1} should have indices from the set $\{k_1,k_2,\ldots,k_{d_2}\}$  which are not present in the set $\{i_1,i_2,\ldots,i_{d_1}\}$. Similarly, The $k-d_2$ vacant positions in Figure \ref{fig2} should have indices from the set $\{i_1,i_2,\ldots,i_{d_1}\}$  which are not present in the set $\{k_1,k_2,\ldots,k_{d_2}\}$. But, we have $\{i_1,i_2,\ldots,i_{d_1}\}$ and $\{k_1,k_2,\ldots,k_{d_2}\}$ differ in atleast $d^*$ indices ($d^*$ is the minimum Hamming distance of $\mathbf{H}$). Hence, we have
\begin{align*}
(k-d_1)+(k-d_2)\geq d^*.
\end{align*} 
This implies
\begin{align*}
2k\geq d^*+d_1+d_2 \geq 3d^*.
\end{align*}
This is a contradiction because we assumed $\frac{3d^*}{2} > k$. Hence, there exists no $k$ column set which are common in both $S_1$ and $S_2$.

\textbf{Case (ii)} $k < n-k$:- 

In Lemma \ref{lemma2}, we proved that the number of linearly dependent $n-k$ column sets in $\mathbf{H}$ are same as that of number of linearly dependent $k$ column sets in $\mathbf{G}^{(s)}$. Let $\mathbf{H}^{(s)}=[\mathbf{I}_{(n-k) \times (n-k)}~\mathbf{P}_{n-k \times k}^\mathsf{T}]$. That is, $\mathbf{H}^{(s)}$ is obtained by rearranging the columns of $\mathbf{H}$. The number of linearly dependent column sets in $\mathbf{H}$ and $\mathbf{H}^{(s)}$ are same. The matrix $\mathbf{H}^{(s)}$ satisfies the condition given in \textbf{case (i).}  
Hence, for $\mathbf{H}^{(s)}$, we can write \eqref{dbound1} as 
\begin{align}
\label{dbound2}
\mathcal{D}=\sum_{d=1}^{k} A_d{{n-d}\choose{n-k}},
\end{align}
where $A_d$s are the weight enumerating coefficients of $(\mathfrak{C}^\mathsf{T})^\mathsf{T}=\mathfrak{C}$. 

The number of $k$ column sets that are linearly dependent is equal to the number of $k \times k$ singular matrices in $\mathbf{G}$. This completes the proof.
\end{proof}

\begin{theorem}
\label{thm2}
Consider an arbitrary binary full rank matrix $\mathbf{G}$ of size $k \times n~(n \geq k)$. Let $\mathfrak{C}$ be the linear block code generated by $\mathbf{G}$ as generator matrix and $\mathfrak{C}^\mathsf{T}$ be the dual code of $\mathfrak{C}$. Let $\mathcal{I}$ be the number of full rank submatrices of size $k \times k$ in $\mathbf{G}$. If 
$\frac{3d^*}{2} > \text{max}(k,n-k)$, then 
\begin{align*}
\mathcal{I} = {{n}\choose{k}}-\sum_{d=1}^kA_d{{n-d}\choose{n-k}},
\end{align*}
where 
\begin{itemize}
\item $A_d$s are the WECs of $W_{\mathfrak{C}^\mathsf{T}}(x,y)$ if $k \geq n-k$
\item $A_d$s are the WECs of $W_{\mathfrak{C}}(x,y)$ if  $k < n-k$.
\end{itemize}
\end{theorem}
\begin{proof}
In a binary matrix of size $k \times n$, there exist ${n}\choose{k}$ binary sub-matrices of size $k \times k$. Each of this sub-matrix may be of full rank or singular. Thus, we have 
\begin{align*}
\mathcal{D}+\mathcal{I}={{n}\choose{k}}
\end{align*}
where $\mathcal{D}$ is the number of singular $k \times k$ sub-matrices in $\mathbf{G}$. The remaining proof follows from Theorem \ref{s1thm1}.
\end{proof}
\begin{remark}
To calculate $\mathcal{D}$ and $\mathcal{I}$, we use the weight enumerating coefficients of either $\mathfrak{C}$ or $\mathfrak{C}^{\mathsf{T}}$ that satisfy the condition $$\text{dimension}< \left \lceil \frac{\text{length of the code}}{2} \right \rceil.$$
\end{remark}
\begin{remark}
For a given binary full rank matrix of size $k \times n$, to calculate  $\mathcal{I}$ by using brute-force technique requires the rank computation of $n \choose k$ number of $k \times k$ matrices. Whereas, by using Theorem \ref{thm2}, one only needs to know the weight enumerating coefficients of the given $k \times n$ matrix or its orthogonal complement. The weight enumerating coefficients required in Theorem \ref{thm2} of any arbitrary $k \times n$ matrix needs the computation of Hamming weight of $2^{\text{min}(k,n-k)}$ binary vectors of size $n$.
\end{remark}

For a binary full rank matrix $\mathbf{G}$ of size $k \times n$ satisfying the condition $\frac{3d^*}{2}>\text{max}(k,n-k)$, Algorithm \ref{algo1} finds the number of full rank submatrices of size $k \times k$ in $\mathbf{G}$.  

		\begin{algorithm}[h]
		\caption{Algorithm to find the number of full rank submatrices of size $k \times k$ in a given full rank matrix $\mathbf{G}$ of size $k \times n$ }
			\begin{algorithmic}[2]
			
				\item [Step 1]~~~
				\begin{itemize}
				\item[\footnotesize{1.1:}] Convert the matrix $\mathbf{G}$ in symmetric form by using elementary row operations. Let this symmetric matrix be $\mathbf{G}^{(s)}=[\mathbf{I}_{k \times k}:~\mathbf{P}_{k \times (n-k)}]$ (Lemma \ref{lemma1}).
				\item[\footnotesize{1.2:}] If $k < \lceil \frac{n}{2} \rceil$,  Go to Step 2.
				\item[\footnotesize{1.3:}] Else $\mathbf{G}^{(s)}=[\mathbf{P}_{(n-k) \times k}^{\mathsf{T}}:~\mathbf{I}_{(n-k) \times (n-k)}]$ (Lemma \ref{lemma2}). 
				\end{itemize}

				\item [Step 2]~~~
				\begin{itemize}
				\item[\footnotesize{2.1:}] Find the weight enumerating coefficients of $\mathbf{G}^{(s)}$. Let $A_d$ for $d \in [1:n]$ be the weight enumerating coefficients of $\mathbf{G}^{(s)}$. 
				\item[\footnotesize{2.2:}] $\mathcal{I} = {{n}\choose{k}}-\sum_{d=1}^kA_d{{n-d}\choose{n-k}}$ (Theorem \ref{thm2})
				\end{itemize}
				\item [Step 3] Exit.
		
			\end{algorithmic}
			\label{algo1}
		\end{algorithm}
		
\begin{example}
\label{ex7}
Consider matrix $\mathbf{G}$ given below.
\arraycolsep=0.0pt
\setlength\extrarowheight{-4.0pt}
{
$$\mathbf{G}=\left[\begin{array}{*{20}c}
   1~1~1~0~1~0~0  \\
   1~0~1~1~0~0~1 \\
   1~1~1~1~1~1~1 \\
   0~1~1~0~0~1~1
   \end{array}\right].$$
   } 

For the matrix $\mathbf{G}$, we have $k=4,n=7$. By using the elementary row operations, the matrix $\mathbf{G}$ can be converted into symmetric matrix $\mathbf{G}^{(s)}$. the matrix $\mathbf{G}^{(s)}$ is given below.
\arraycolsep=0.0pt
\setlength\extrarowheight{-4.0pt}
{
$$\mathbf{G}^{(s)}=[\mathbf{I}_{4 \times 4}:~\mathbf{P}_{4 \times 3}]=\left[\begin{array}{*{20}c}
   1~0~0~0~1~1~1  \\
   0~1~0~0~1~1~0 \\
   0~0~1~0~1~0~1 \\
   0~0~0~1~0~1~1

   \end{array}\right].$$
   } 

For the matrix $\mathbf{G}^{(s)}$, we have $k \geq n-k$. Hence, we need to calculate $\mathcal{D}$ and $\mathcal{I}$ from weight enumerating coefficients of $\mathfrak{C}^{\mathsf{T}}$. The generator matrix of $\mathfrak{C}^{\mathsf{T}}$ is given below.
\arraycolsep=0.0pt
\setlength\extrarowheight{-4.0pt}
{
$$[\mathbf{P}^{\mathsf{T}}_{3 \times 4}:~\mathbf{I}_{3 \times 3}]=\left[\begin{array}{*{20}c}
   1~1~1~0~1~0~0 \\
   1~1~0~1~0~1~0 \\
   1~0~1~1~0~0~1 
   \end{array}\right].$$
   }

The weight enumerating function of $\mathfrak{C}^{\mathsf{T}}$ can be calculated by enumerating all the $2^3=8$ codewords generated by above $3 \times 7$ matrix. The weight enumerating function is given by 
\begin{align}
\label{ex2wef}
\mathcal{W}_{\mathfrak{C}^{\mathsf{T}}}(x,y)=x^7+7x^4y^4.
\end{align}

From \eqref{ex2wef}, we have $d^*=4,A_4=7$ and $A_d=0$ for $d \neq 4$. The $k$ and $d^*$ in this example satisfy the condition $\frac{3d^*}{2}>\text{max}(k,n-k)$. From Theorem \ref{s1thm1}, we have 
\begin{align*}
\mathcal{D}=\sum_{d=1}^{k} A_d{{n-d}\choose{n-k}}=A_4 {n-4 \choose n-k}=7{4 \choose 4}=7.
\end{align*}
From Theorem \ref{thm2}, we have
$\mathcal{I}={n \choose k}-\mathcal{D}=35-7=28.$

The 4-column sets corresponding to the $7$ singular $4 \times 4$ matrices are given below.
\begin{align*}
&\mathcal{D}_1=\{1,2,3,5\},~~~\mathcal{D}_2=\{1,2,4,6\},~~~\mathcal{D}_3=\{1,3,4,7\}\\&\mathcal{D}_4=\{3,4,5,6\},~~~\mathcal{D}_5=\{2,3,6,7\},
~~~\mathcal{D}_6=\{2,4,5,7\}\\&\mathcal{D}_{7}=\{1,5,6,7\}.
\end{align*}

The 4-column sets corresponding to the $28$ full-rank $4 \times 4$ matrices are given below.

\begin{align*}
&\mathcal{I}_1=\{1,2,3,4\},~~~\mathcal{I}_2=\{1,2,3,5\},~~~\mathcal{I}_3=\{1,2,3,7\}\\&\mathcal{I}_4=\{1,2,4,6\},~~~\mathcal{I}_5=\{1,2,4,7\},
~~~\mathcal{I}_6=\{1,2,5,6\}\\&\mathcal{I}_7=\{1,2,5,7\},~~~\mathcal{I}_8=\{1,2,6,7\},~~~\mathcal{I}_9=\{1,3,4,5\},\\&\mathcal{I}_{10}=\{1,3,4,6\},~~\mathcal{I}_{11}=\{1,3,5,6\},~~\mathcal{I}_{12}=\{1,3,5,7\},\\&
\mathcal{I}_{13}=\{1,3,6,7\},~~\mathcal{I}_{14}=\{1,4,5,6\},~~\mathcal{I}_{15}=\{1,4,5,7\},\\&\mathcal{I}_{16}=\{1,4,6,7\},~~\mathcal{I}_{17}=\{2,3,4,5\},~~\mathcal{I}_{18}=\{2,3,4,6\},\\&\mathcal{I}_{19}=\{2,3,4,7\},~~\mathcal{I}_{20}=\{2,3,5,6\},~~
\mathcal{I}_{21}=\{2,3,6,7\},\\&\mathcal{I}_{22}=\{2,4,5,6\},~~\mathcal{I}_{23}=\{2,4,5,7\},~~\mathcal{I}_{24}=\{2,5,6,7\},\\&\mathcal{I}_{25}=\{3,4,5,7\} ,~~\mathcal{I}_{26}=\{3,4,6,7\},~~\mathcal{I}_{27}=\{3,5,6,7\}\\&\mathcal{I}_{28}=\{4,5,6,7\}.
\end{align*}

\end{example}
\begin{note}
The matrix $\mathbf{G}$ given in Example \ref{ex7} is a generator matrix of $(7,4)$ Hamming code. Hence, in the generator matrix of $(7,4)$ Hamming code, there exist 28 full rank submatrices of size $4 \times 4$. From Lemma \ref{lemma2}, in the parity check matrix of $(7,4)$ Hamming code, there exist 28 full rank submatrices of size $3 \times 3$. 
\end{note}

\begin{example}
\label{ex8}
Consider matrix $\mathbf{G}$ given below.
\arraycolsep=0.0pt
\setlength\extrarowheight{-4.0pt}
{
$$\mathbf{G}=\left[\begin{array}{*{20}c}
   1~1~1~1~1~1~0~0~1~0 \\
   1~0~1~0~0~0~0~0~0~0  \\
   0~1~0~0~0~0~0~0~1~1 \\
   0~0~1~0~0~0~0~1~1~0 \\
   0~0~0~0~1~0~0~0~1~0 \\
   1~0~1~0~0~0~1~1~0~1 \\
   0~0~0~1~0~0~0~1~0~1  \\

   \end{array}\right].$$
   } 

For the matrix $\mathbf{G}$, we have $k=7,n=10$ and the matrix $\mathbf{G}^{(s)}$ is given below.
\arraycolsep=0.0pt
\setlength\extrarowheight{-4.0pt}
{
$$\mathbf{G}^{(s)}=[\mathbf{I}_{7 \times 7}:~\mathbf{P}_{7 \times 3}]=\left[\begin{array}{*{20}c}
   1~0~0~0~0~0~0~1~1~0  \\
   0~1~0~0~0~0~0~0~1~1 \\
   0~0~1~0~0~0~0~1~1~1 \\
   0~0~0~1~0~0~0~1~0~1  \\
   0~0~0~0~1~0~0~0~1~0 \\
   0~0~0~0~0~1~0~1~1~1 \\
   0~0~0~0~0~0~1~1~0~1

   \end{array}\right].$$
   } 

For the matrix $\mathbf{G}^{(s)}$, we have $k \geq n-k$. Hence, we need to calculate $\mathcal{D}$ and $\mathcal{I}$ from weight enumerating coefficients of $\mathfrak{C}^{\mathsf{T}}$. The generator matrix of $\mathfrak{C}^{\mathsf{T}}$ is given below.
\arraycolsep=0.0pt
\setlength\extrarowheight{-4.0pt}
{
$$[\mathbf{P}^{\mathsf{T}}_{3 \times 7}:~\mathbf{I}_{3 \times 3}]=\left[\begin{array}{*{20}c}
   1~0~1~1~0~1~1~1~0~0 \\
   1~1~1~0~1~1~0~0~1~0 \\
   0~1~1~1~0~1~1~0~0~1 
   \end{array}\right].$$
   } 

The weight enumerating function of $\mathfrak{C}^{\mathsf{T}}$ can be calculated by enumerating all the $2^3=8$ codewords generated by above $3 \times 10$ matrix. The weight enumerating function is given by 
\begin{align}
\label{ex21wef}
\mathcal{W}_{\mathfrak{C}^{\mathsf{T}}}(x,y)=x^{10}+7x^4y^6.
\end{align}

From \eqref{ex21wef}, we have $d^*=6,A_6=7$ and $A_d=0$ for $d \neq 6$. The $k$ and $d^*$ in this example satisfy the condition $\frac{3d^*}{2}>\text{max}(k,n-k)$. From Theorem \ref{s1thm1}, we have 
\begin{align*}
\mathcal{D}&=\sum_{d=1}^{k} A_d{{n-d}\choose{n-k}}\\&=A_6 {n-6 \choose n-k}=6{4 \choose 3}=28.
\end{align*}
From Theorem \ref{thm2}, we have
\begin{align*}
\mathcal{I}={n \choose k}-\mathcal{D}=120-28=92.
\end{align*}
\end{example}
\begin{example}
Consider a generator matrix $\mathbf{G}$ of $(n=15,k=11)$ Hamming code. The weight enumerating function of the dual of the $(15,11)$ Hamming code is 
\begin{align*}
\mathcal{W}_{\mathfrak{C}^{\mathsf{T}}}(x,y)=x^{15}+15x^{7}y^{8}.
\end{align*}
In this example, we have $k \geq n-k$ and hence $d^*$ is the minimum Hamming weight of the error correcting code generated by $\mathbf{H}$ and $d^*=8$. The $k$ and $d^*$ in this example satisfy the condition $\frac{3d^*}{2}>\text{max}(k,n-k)$. From Theorem \ref{s1thm1}, we have 
\begin{align*}
\mathcal{D}&=\sum_{d=1}^{k} A_d{{n-d}\choose{n-k}}\\&=A_{8} {n-8 \choose n-k}=15{7 \choose 4}=525.
\end{align*}
From Theorem \ref{thm2}, we have
\begin{align*}
\mathcal{I}={15 \choose 11}-\mathcal{D}=1,365-525=840.
\end{align*}
\end{example}
\begin{note}
In a generator matrix of $(15,11)$ Hamming code, there exist 840 full rank submatrices of size $11 \times 11$. From Lemma \ref{lemma2}, in the parity check matrix of $(15,11)$ Hamming code, there exist 840 full rank submatrices of size $4 \times 4$. 
\end{note}

\section{scope for future work}
\label{sec5}
In this paper, for a $k \times n$ binary full rank matrix satisfying a condition of minimum Hamming weight, we quantified the number of binary full rank submatrices of size $k \times k$  in terms of weight enumerating function. 
Some of the interesting problems are given below.
\begin{itemize}
\item Establishing the relation between the number of full rank submatrices of size $k \times k$ for any arbitrary full rank binary matrix of size $k \times n$ is an interesting problem.
\item For given positive integers $k$ and $n$ ($n > k$), finding the maximum value of $\mathcal{I}$ and the corresponding matrix $\mathbf{G}$ is an interesting problem.
\item Extending the result given in this paper to other fields is an interesting problem.
\end{itemize}
\section*{Acknowledgement}
This work was supported partly by the Science and Engineering Research Board (SERB) of Department of Science and Technology (DST), Government of India, through J.C. Bose National Fellowship to B. Sundar Rajan.


\begin{thebibliography}{9}
\bibitem{NC1}
R. Ahlswede, N. Cai, R. Li, and R. Yeung, “Network information flow,” in proc.  \textit{IEEE Trans. Inf. Theory}, vol. 46, no. 4, pp. 1204–1216, Jul. 2000.
\bibitem{NC2}
N. Cai, S.-Y. R. Li, and R. W. Yeung, “Linear network coding,” in proc. \textit{IEEE Trans. Inf. Theory}, vol. 49, no. 2, pp. 371–381, Feb. 2003.
\bibitem{NC3}
R. Koetter and M. Medard, “An algebraic approach to network coding,” in proc. \textit{IEEE/ACM Trans. Netw.}, vol. 11, no. 5, pp. 782–795, Oct. 2003.
\bibitem{NC4}
S. Jaggi et al., “Polynomial time algorithms for multicast network code construction,” in proc. \textit{IEEE Trans. Inf. Theory}, vol. 51, no. 6, pp. 1973–1982, Jun. 2005.
\bibitem{ISCO}
Y. Birk and T. Kol, ``Informed-source coding-on-demand (ISCOD) over broadcast channels", in \textit{Proc. IEEE Conf. Comput. Commun.}, San Francisco, CA, 1998, pp. 1257-1264.
\bibitem{RY}
R. Yeung, S. Y. Li and  N Cai, ``Network Coding Theory", Now Publishers Inc 2006.
\bibitem{ECC2}
F. J. MacWilliams, ``A theorem on the distribution of weights in a systematic codes", Bell Syst. Tech. J. 42: 79-94, 1963.

\bibitem{ECC}
F. J. MacWilliams and N. J. Sloane,  ``The theory of error correcting codes", Elsevier, 1977.
\end{thebibliography}
\end{document}